\documentclass{article}
\usepackage[margin=1in]{geometry}
\usepackage{latexsym}

\usepackage{graphicx,subcaption}
\usepackage{mathptmx}
\usepackage[utf8]{inputenc}
\usepackage[english]{babel}

\usepackage{amsmath}
\usepackage{amsfonts}
\usepackage{amssymb}
\usepackage{amsbsy}
\usepackage{amsthm}

\newtheorem{theorem}{Theorem}
\newtheorem{lemma}{Lemma}

\newtheorem{assumption}{Assumption}
\usepackage{personal}
\date{}

\author{Prateek Jaiswal$^1$\and Harsha Honnappa$^1$\and Raghu Pasupathy$^2$\\}
\date{
	$^1$School of Industrial Engineering, Purdue University\\
	$^2$Department of Statistics, Purdue University\\[2ex]
}

\begin{document}

\title{\uppercase{Optimal Allocations for Sample Average Approximation}}

\maketitle

\vspace{1em}
\section*{Abstract}
 We consider a single stage stochastic program without recourse with a strictly convex loss function. We assume a compact decision space and grid it with a finite set of points. In addition, we assume that the decision maker can generate samples of the stochastic variable independently at each grid point and form a \textit{sample average approximation} ({\sc SAA}) of the stochastic program. Our objective in this paper is to characterize an asymptotically optimal linear sample allocation rule, given a fixed sampling budget, which maximizes the decay rate of probability of making false decision. 
 
\section{Introduction}
\label{sec:intro}
Let $\xi$ be a
measurable function that induces the distribution function $F(x) :=
\mathbb P(\xi \leq x)$. Stochastic programs are canonical models for
decision-making under uncertainty, covering a multitude of classic
stochastic optimization problems:
\begin{align}
\label{eq:sp}
&\text{minimize} \quad f(x) := \mathbb E\left[ L(x,\xi)\right] = \int_\Xi
L(x,\xi) dF(\xi),\\
\nonumber 
&\text{subject to} \quad x \in \mathcal X,
\end{align}
where $L(x,\xi) : \mathcal{X}\times \Xi \rightarrow {\rm I\!R}$ is
a continuously differentiable `loss' function, $\mathcal{X} \subset
{\rm I\!R}$ is the set of decision variables; for brevity we also call the sample values
as $\xi$. In this paper, we assume that $L(x,\xi)$ is convex in $x$
and measurable with respect to $\xi$.

\textit{Sample average approximation} ({\sc SAA}) is a
classic Monte Carlo method for estimating the stochastic
program~\eqref{eq:sp}.  Here, the decision maker (DM) grids the decision
space $\mathcal{X}$ into a finite set of points $D := \{x_1, x_2,
..., x_{d}\}$, and simulates samples of the loss function
$L(x_i,\xi)$ at each of the grid points. Assuming a total sampling
budget of $n$, the DM generates $m(x_i)$ independent and identically distributed (i.i.d) samples of $L(x_i,\xi)$, denoted as $L(x_i,\xi^i_j) ~\forall ~ 1\leq j \leq m(x_i)$; at each point $x_i$ with $\sum_{i=1}^{d} m(x_i) = n$. Then, 
the {\sc SAA} stochastic program is
\begin{align}
\label{eq:saa}
&\min_{x_i \in D} \quad \hat f(x_i) := \frac{1}{m(x_i)}\sum_{j=1}^{m(x_i)} L(x_i,\xi_j^i).
\end{align}

Modulo regularity conditions on $L(\cdot,\cdot)$, the strong law of
large numbers (SLLN) implies that~\eqref{eq:saa} converges to the
true program in~\eqref{eq:sp} ~\cite{ShDeRu2009}. It is also known that
the optimizers are consistent and the optimal rate of convergence is
$O(n^{-1/2})$. An important allied question to these asymptotic
results is a quantification of the likelihood that the empirical
optimizer diverges from the true optimizer for a given sampling budget
$n$. We seek such a quantification for two reasons:
\begin{itemize}
	\item A quantification of this rate gives a clear sense of how `good'
	the empirical optimizer and empirical optimal value are, and 
	\item in the simulation context it provides a
	guideline on how to allocate a limited
	sampling/computational budget across the design points in $D$.
\end{itemize}

In this paper, we focus on the latter issue. In general, it is a formidable task to compute the likelihood for a
fixed budget; in full generality, one requires tight concentration
bounds in order to make meaningful predictions about budget allocations. Instead, in this paper we establish a large deviations.
principle ({\sc LDP}) satisfied by the Monte Carlo estimator~\eqref{eq:saa}
as the sampling budget tends to infinity.
For the definition of LDP, we refer the readers to section 1.2 of \cite{DeZe2010}.

It is important here to differentiate between the optimally computing budget allocation (OCBA) method for selecting an optimal system from a finite set of systems \cite{Chen2000,GlJu2004} and our approach to {\sc SAA} problem. In OCBA, there is no topology associated with the finite set of systems unlike {\sc SAA}. In addition, OCBA approach only considers  probability of selecting suboptimal system due to random sampling errors, whereas our framework also takes into account the discretization error.  

\subsection{Our Contributions} 
Our main objective is the derivation of an optimal allocation of the
sampling budget across the design points such that a canonical LD rate
is achieved at the optimizer of~\eqref{eq:saa}. In particular, we seek
what we term as `linear' allocation rules where $m(x) = \a_x n$, where
$\a_x \in [0,1]$ and $\sum_{x \in D} \a_x = 1$. 

We make the simplifying assumption that the DM can sample
independently from each design point $x \in D$. In effect, this allows
an `embarrassingly' parallel implementation of the {\sc SAA} estimation,
where `slave' machines compute $\hat f(x)$ with $m(x)$ samples, and communicate the
result to a central `master' machine that coordinates the budget
allocation and aggregates the calculations to compute~\eqref{eq:saa}.  

 Now let $\hat{x}_n := \arg\min_{x_i \in D} \hat{f}(x_i)$
, $\hat{x} := \arg\min_{x_i \in D} {f}(x_i)$
, and ${x^*} :=
\arg\min_{x \in
	\mathcal{X}}
{f}(x)$. As noted before, the Monte Carlo empirical objective converges to the
population objective as $n \to \infty$ for any linear allocation rule, and consequently one expects that $\hat
x_n \to \hat x$ almost surely (a.s.) as $n \to \infty$. Since the objective is assumed to
be continuously differentiable, it follows that $f(\hat x_n) \to
f(\hat x)$ a.s. as $n\to\infty$.  With this information, we establish the
canonical LD rate function satisfied by the optimal value
of~\eqref{eq:saa} as a function of some linear allocation rule and in
the limit of a large sampling budget. This
result follows from the G\"artner-Ellis theorem, and does not involve
any analytical subtleties in light of our assumptions. However, the
rate function has not appeared in the literature before, and might be
of interest more generally. The proof proceeds in two steps. We first
characterize the LD rate of the likelihood of mis-ordering 
the {\sc SAA} empirical estimates at any two points in the decision space
$D$. Next, we use this result to establish our main result on
the LD rate on the likelihood that the objective value $f(\hat x_n)$
at the {\sc SAA} optimizer $\hat x_n$
in~\eqref{eq:saa} is at least $\epsilon > 0$ worse than the true
value $f(x^*)$, in the large budget limit.

We next provide structural results on the LD rate function, in
particular demonstrating that it is strictly concave in the allocation
rule $\mathbf \a := (\a_x,~x\in D)$. Consequently, there exists a
unique optimal linear allocation rule. This result, of course,
presumes that the master machine has complete
information about the statistics of the expected loss function - in
particular, we assume the existence of a cumulant generating
function. In practice, this is not an implementable policy, since the
DM only has access to a Monte Carlo simulator. We next design two
recursive algorithms that optimize the LD rate function as samples
accumulate. The first algorithm parallels Algorithm 2
in~\cite{HuPa2013} and is applicable when a closed form expression for
the rate function for the mis-ordering likelihood is available. When such an
expression is available the problem is really one of ranking and
selection (R\&S). In general {\sc SAA} problems, closed-forms are not easy
to compute and this too must be estimated. Our second algorithm is an
`expectation-maximization' style recursive algorithm. We illustrate
these algorithms with numerical simulation results.

The remainder of the paper is organized as follows. We begin in
Section~\ref{sec:ld} by proving the {\sc LDP} satisfied by the {\sc SAA}
estimator. In Section~\ref{sec:ola} we derive structural properties of
the LD rate function as function of the linear allocation rule, and
exhibit the variational optimization problem to find the optimal
linear allocation rule. We then derive a recursive algorithm for
computing the optimal allocation rule on a sample path (and fixed
sampling budget), and illustrate the algorithm, on three different
example problems. We end with comments on several future directions
for this paper.

\section{Notations and Preliminaries}
We assume there exists a probability sample space $(\Omega, \Xi,
\mathbb P)$, and define $\xi$ with respect to this space.  The
indicator function
of a set is
represented 
by $\mathbb{I}_{\{\cdot\}}$ and $\lfloor \cdot \rfloor$ denotes the greatest integer function. We define a `regret' function
\begin{align}\label{eq:OR}
	\left[ f(\hat{x}_n) - f(x^*) \right] = \left[ f(\hat{x}_n) - f(\hat{x}) \right] + \left[ f(\hat{x}) - f(x^*) \right],
\end{align}
where the first term on the right hand side is the \textit{Sampling
  Error} and the latter term the \textit{Discretization Error}. 
We make a few assumptions to guarantee the existence
of the LD rate function. We assume that the loss function $L(\cdot,\cdot)$ satisfies

\begin{assumption}
	\label{assume:2}
	$L(x,\xi)$ is not a point mass at $f(x)$ for all $x \in \sX$ and for
	some continuously differentiable function $f$.
\end{assumption}

\begin{assumption}
	\label{assume:3} 
	The cumulant generating function (CGF) of $L(\cdot, \xi)$ is well
	defined and finite for all $x \in \mathcal{X}$, that is
	\begin{align*}
	\Lambda(x,\theta):= \log \mathbb{E} \left[ e^{\theta L(x,\xi)} \right] < \infty \quad \forall \theta \in \mathbb{R}, \quad \forall x \in \mathcal{X} .
	\end{align*}
\end{assumption}

When the loss function $L(\cdot, \xi)$ is bounded above by $\xi$ and the CGF of $\xi$ is well defined and finite, the above assumption is trivially satisfied. 

\begin{assumption}
	\label{assume:4}
	Let $H_x(\Lambda):=\{\theta: \Lambda(x,\theta) < \infty\} \forall \
	x \in \mathcal{X}$ be such that the origin belongs to the interior
	of $H_x(\Lambda) $. Furthermore, we also assume that $\Lambda(x,\theta)$ is steep, that is $\lim_{n\to \infty}\left \vert \frac{\partial \Lambda(x,\theta) }{\partial \theta} \right\vert_{\theta= \theta_n} = \infty$, for any sequence $\theta_n$ in the interior of  $H_x(\Lambda)$, which converges to a boundary point of $H_x(\Lambda)\ \forall \ x \in \mathcal{X} $. 
\end{assumption}

\section{Large Deviations for {SAA}} 
\label{sec:ld}

In this section we establish an {\sc LDP} satisfied by the regret
function. Let $x,y \in D$ and consider the Monte Carlo estimates $\hat f(x)$ and
$\hat f(y)$. Our first result establishes a LD rate function for the
likelihood that $\hat f(x)$ and $\hat f(y)$ are mis-ordered in the
large budget limit.
 
\begin{lemma}\label{lem:LDhat}
	Fix $\gamma > 0$ and $x,y \in D$. Then, under
        Assumptions~\ref{assume:2},~\ref{assume:3} and~\ref{assume:4},
	\begin{align}
		\nonumber
		\lim_{n\to \infty} \frac{1}{n} \log \mathbb{P}(  \hat f(y)- \hat f(x) \geq \gamma )
		&= - I(\gamma,\alpha_x,\alpha_y) \mathbb{I}_{\{ f(y)-f(x) < \gamma \}},
	\end{align}
	where  
	$I(\gamma,\alpha_x,\alpha_y):= \sup_{t \in R} \left( t \gamma
          - \alpha_y \Lambda \left(y,\frac{t}{\alpha_y} \right)  - \alpha_x \Lambda\left(x,-\frac{t}{\alpha_x}\right)  \right) $.
\end{lemma}
\begin{proof}
	
	Let $Y_n := n \left( \hat{f}(y)- \hat{f}(x) \right)$, and observe that for any $t \in \mathbb{R}$
	\begin{align*}
		\mathbb E[e^{tY_n}] &= \mathbb E \left[ \exp\left({t\sum_{i=1}^{\lfloor \alpha_y n \rfloor} \frac{L(y,\xi_i^y)}{\alpha_y} - t\sum_{j=1}^{\lfloor \alpha_x n \rfloor} \frac{L(x,\xi_j^x)}{\alpha_x}}\right) \right] 
		\\
		&= \mathbb E \left[ \exp\left({t \frac{L(y,\xi)}{\alpha_y} }\right) \right]^{\lfloor \alpha_y n \rfloor} \mathbb E \left[ \exp\left({- t\frac{L(x,\xi)}{\alpha_x}}\right) \right]^{\lfloor \alpha_x n \rfloor},
	\end{align*}
	where the last equality follows from the fact that we
        sample independently at every design point. Next, using the fact that $\lim_{n \to \infty} \lfloor \alpha_i n \rfloor / n = \alpha_i $ 
observe that 
\[
\lim_{n \to \infty} \frac{1}{n} \log\mathbb E\left[e^{tY_n} \right] =  \alpha_y \log  \mathbb E \left[ \exp\left({t
    \frac{L(y,\xi)}{\alpha_y} }\right) \right] + \alpha_x \log \mathbb E
\left[ \exp\left({- t\frac{L(x,\xi)}{\alpha_x}} \right) \right] =: 
\varphi(t,\alpha_x,\alpha_y) \quad \forall t \in
\mathbb{R}.
\]
By Assumption ~\ref{assume:3} , $\varphi(t,\alpha_x,\alpha_y) <
\infty \quad \forall t \in  \mathbb{R} $. Together with
Assumption~\ref{assume:4} it follows that the G\"artner--Ellis Theorem holds ~\cite{DeZe2010}, and the lemma is proved with
good rate function $I(z,\alpha_x,\alpha_y):= \sup_{t \in R} \left( tz
  - \varphi(t,\alpha_x,\alpha_y)\right) $. Since $I(z,\alpha_x,\alpha_y)$ is strictly
convex in $z$ and attains the minimum value 0 precisely at $f(y)-f(x)$, 	
Therefore, $\inf_{z \in [\gamma,\infty)} I(z,\alpha_x,\alpha_y)= I(\gamma,\alpha_x,\alpha_y) \mathbb{I}_{\{ f(y)-f(x) < \gamma \}}. $	
\end{proof}

The next two lemmas
are crucial for establishing the
main result of this section.

\begin{lemma}\label{lem:limmax}
	Fix $k \in
        \mathbb{N}$
        and let
        $\{a^n_i\}
        \subset \bbR$
        be 
        arbitrary
        sequences
        for $~1\leq
        i \leq k$. Then
	\[
        \liminf_{n
          \to
          \infty}
        \max\{a_1^n,a_2^n,
        \ldots,
        a_k^n \}
        \geq \max\{
        \liminf_{n
          \to
          \infty}
        a_1^n,
        \liminf_{n
          \to
          \infty}
        a_2^n,
        \ldots,
        \liminf_{n
          \to
          \infty}a_k^n\}. 
        \]
	\end{lemma}
\begin{proof}
	First observe that for any $i \in \{1,2,\ldots,k\}$
\(
		\max\{a_1^n,a_2^n, \ldots, a_k^n \} \geq a^n_i.	
\)
	~Therefore, 
\(
	\liminf_{n \to \infty} \max\{a_1^n,a_2^n, \\ \ldots, a_k^n \} \geq  \liminf_{n \to \infty} a^n_i.	
\)
	~Since this holds for any $i \in \{1,2,\ldots,k\}$, the lemma follows. 
	\end{proof}

        \begin{lemma}\label{lem:limsum}
          Let $\{a_i^n\} \subset \bbR$ for $i=1,2$ be arbitrary
          sequences. Then,
          \[
          \liminf_{n \to
            \infty}
          (a_1^n +
          a_2^n)
          \geq
\liminf_{n
            \to
            \infty}
          a_1^n  +
          \liminf_{n
            \to
            \infty}
          a_2^n. 
          \]
	\end{lemma}
\begin{proof}
	The proof follows from the definition of $\liminf$ and infimum inequality.
	
	\end{proof}

 We now turn to main result, which establishes an {\sc LDP} for the regret~\eqref{eq:OR}. 
 
\begin{theorem} \label{thrm:1}
	Fix $\e > f(\hat{x}) - f(x^*)$. Under
        Assumptions \ref{assume:2}, \ref{assume:3} and
        \ref{assume:4} the regret~\eqref{eq:OR}~satisfies
        \begin{align}
        \nonumber
        \lim_{n\to \infty} \frac{1}{n} \log \mathbb{P}(  f(\hat{x}_n) - f(x^*) \geq \e )
        &= - J(\e),
        \end{align}
          
	where  
	$J(\e) := \min_{x \in Q(\d)}  \sum _{y \in D} I(\alpha_x,\alpha_y) \mathbb{I}_{\{ f(y) < f(x) \}}$,
	$I(\alpha_x,\alpha_y):= \sup_{t \in R} \left( - \alpha_y \Lambda(y,\frac{t}{\alpha_y})  - \alpha_x \Lambda(x,-\frac{t}{\alpha_x})  \right) $, and $Q(\d):= \{x \in D: f(x)  > f(x^*) + \epsilon \}$.

\end{theorem}

\begin{proof}
Recall that our objective is to characterize the rate of decay of the likelihood of the rare event,
	\begin{align}
          \{
           f(\hat{x}_n)- f(x^*) > \epsilon
          \}.
	\end{align}	
	Observe that
	\begin{align*}
	\mathbb{P}(  f(\hat{x}_n)- f(x^*) > \epsilon ) &= \mathbb{P} \left(  f(\hat{x}_n) - f(\hat{x}) + f(\hat{x}) - f(x^*) > \epsilon \right)\\
	&= \mathbb{P} \left(  f(\hat{x}_n) - f(\hat{x})  > \delta) \right),  	
	\end{align*}
	where $\delta := \epsilon - (f(\hat{x}) - f(x^*))$, and
        $f(\hat{x}) - f(x^*)$ is the non-random discretization
        error. From the definition of $Q(\d)$ set
		observe that the event $\left\{  f(\hat{x}_n) - f(\hat{x})
          > \delta) \right\}$ is equivalent to $\left\{  \hat{x}_n \in
          Q(\d) \right\}$. Now, using the definition of $\hat{x}_n$,
        observe the equivalence
	\begin{align*}
	\left\{  \hat{x}_n \in Q(\d) \right\} = \bigcup_{x \in Q(\d)} \bigcap_{y \in D} \left\{ \hat{f}(x) \leq \hat{f}(y) \right\}.
	\end{align*} 
	Therefore, it straightforwardly follows that
	\begin{align}
	\max _{x \in Q(\d)} \mathbb{P} \left( \bigcap_{y \in D} \left\{ \hat{f}(x) \leq \hat{f}(y) \right\}\right) \leq \mathbb{P}(  f(\hat{x}_n)- f(x^*) > \epsilon ) \leq \sum _{x \in Q(\d)} \mathbb{P} \left( \bigcap_{y \in D} \left\{ \hat{f}(x) \leq \hat{f}(y) \right\}\right).
	\label{eq:eq1}
	\end{align}
        Now, since the set $Q(\d)$ is finite, Lemma 1.2.15 of  ~\cite{DeZe2010}
        implies
	\begin{align}
	\nonumber
	\limsup_{n\to \infty} \frac{1}{n} \log \mathbb{P}(  f(\hat{x}_n)- f(x^*) > \epsilon ) &\leq \limsup_{n\to \infty} \frac{1}{n} \log \sum _{x \in Q(\d)} \mathbb{P} \left( \bigcap_{y \in D} \left\{ \hat{f}(x) \leq \hat{f}(y) \right\}\right)
	\\
	\nonumber
	&= \max_{x \in Q(\d)} \limsup_{n\to \infty} \frac{1}{n} \log \mathbb{P} \left( \bigcap_{y \in D} \left\{ \hat{f}(x) \leq \hat{f}(y) \right\}\right)
	\\
	&\leq \max_{x \in Q(\d)}  \sum _{y \in D} \limsup_{n\to \infty} \frac{1}{n} \log \mathbb{P} \left( \hat{f}(x) \leq \hat{f}(y)\right),
	\label{eq:eq2a}
	\end{align} 
where the last inequality follows from the fact that sampling is
independent at each of the design points. Next, for the lower bound, the monotonicity of the logarithm function and the sampling independence implies
\begin{align}
	\nonumber
	\liminf_{n\to \infty} \frac{1}{n} \log \mathbb{P}(  f(\hat{x}_n)- f(x^*) > \epsilon ) &\geq \liminf_{n\to \infty} \frac{1}{n} \log \max _{x \in Q(\d)} \mathbb{P} \left( \bigcap_{y \in D} \left\{ \hat{f}(x) \leq \hat{f}(y) \right\}\right)
	\\
	&\geq  \max _{x \in Q(\d)}   \sum_{y \in D} \liminf_{n\to \infty} \frac{1}{n}  \log  \mathbb{P} \left(  \hat{f}(x) \leq \hat{f}(y)\right),
	\label{eq:eq3a}
	\end{align} 
	where the final inequality follows from
        Lemmas~\ref{lem:limmax} and~\ref{lem:limsum}. Finally, the
        theorem follows from an application of Lemma~\ref{lem:LDhat} to~\eqref{eq:eq2a} and~\eqref{eq:eq3a}.	
\end{proof}

Some comments are in order for this result. First, observe that the
fact that we assume a finite grid $D \subset \sX$ implies that the LD
rate function is well-defined. Furthermore, it can be anticipated that
the rate function can be established even with a countable
grid. Second, the set $Q(\d)$ is critical for establishing the limit. This is the
set of design points in the grid $D$ that are $\e$ worse than the
\textit{optimal} design point $x^* \in \sX$. Thus, the rate function
$J(\e)$ identifies the dominant point on the boundary of the set
$Q(\d)$ that $\hat x_n$ is most likely to diverge away from
$x^*$. The form of $J(\e)$ indicates that it is a composition of the
closest point in $Q(\d)$ to each design point in $D$ that is most
likely to cause a mis-ordering.
\section{Optimal Linear Allocation Rule} \label{sec:ola}

In this section, we identify an (asymptotically) optimal linear
allocation rule that optimizes the rate function in Theorem
~\ref{thrm:1}. Note that this optimization is \textit{post hoc} in the
sense that the rate function is identified for an arbitrary
allocation. The optimal allocation maximizes the rate at which the log
likelihood of misordering the empirical optimizer
approaches zero. First, we show that the rate function obtained in Theorem ~\ref{thrm:1} is {strictly concave} and thus has a unique maximizer. Let $J(\alpha_x,\alpha_y,t):= \left( \alpha_y \Lambda\left(y,\frac{t}{\alpha_y}\right) + \alpha_x \Lambda\left(x,-\frac{t}{\alpha_x}\right)  \right)$ for brevity.

\begin{lemma} \label{lem:concrate1}
	$I(\alpha_x,\alpha_y)= \sup_{t \in R} \left( - \alpha_y
          \Lambda\left(y,\frac{t}{\alpha_y}\right) - \alpha_x
          \Lambda\left(x,-\frac{t}{\alpha_x}\right)  \right)$, is strictly concave $\forall \{\alpha_x,\alpha_y\} \in[0,1]\times [0,1]$.
\end{lemma}

\begin{proof}
	Observe that proving $I(\alpha_x,\alpha_y)$ is concave in
        $\alpha_x$ and $\alpha_y$ is equivalent to proving  \[
        \inf_{t \in R} \left( \alpha_y
          \Lambda\left(y,\frac{t}{\alpha_y}\right) + \alpha_x
          \Lambda\left(x,-\frac{t}{\alpha_x}\right)  \right)  =
        \inf_{t \in R} J(\alpha_x,\alpha_y,t) \] is convex $\forall
        ~\alpha_x,\alpha_y~ \in[0,1]\times [0,1]$. First, we
        demonstrate that $\alpha_y
        \Lambda\left(y,\frac{t}{\alpha_y}\right)$ is convex in $t$ and
        $\alpha_y$. Using the definition of $\Lambda(\cdot,t)$ we have
%
	\begin{align}
	\frac{\partial^2  }{\partial \alpha_y ^2} \alpha_y \Lambda\left(y,\frac{t}{\alpha_y}\right) &= 
	 \frac{t^2}{\alpha_y^3} \left\{  
	  \frac{ \mathbb E \left[ \exp\left( {t
                                                                                                      \frac{L(y,\xi)}{\alpha_y}
                                                                                                      }\right){
                                                                                                      L(y,\xi)}^2
                                                                                                      \right]
                                                                                                      }{
                                                                                                      \mathbb
                                                                                                      E
                                                                                                      \left[
                                                                                                      \exp\left({t
                                                                                                      \frac{L(y,\xi)}{\alpha_y}
                                                                                                      }\right)
                                                                                                      \right]
                                                                                                      }
                                                                                                      -
                                                                                                      \frac{
                                                                                                      \mathbb
                                                                                                      E
                                                                                                      \left[
                                                                                                      \exp\left({t
                                                                                                      \frac{L(y,\xi)}{\alpha_y}
                                                                                                      }
                                                                                                      \right)
                                                                                                      L(y,\xi)
                                                                                                      \right]^2
                                                                                                      }{
                                                                                                      \mathbb
                                                                                                      E
                                                                                                      \left[
                                                                                                      \exp\left({t
                                                                                                      \frac{L(y,\xi)}{\alpha_y}
                                                                                                      }
                                                                                                      \right)
                                                                                                      \right]^2 } \right\} >  0.\label{eq:var}
	\end{align}
	Observe that the expression on the right hand side
        of~\eqref{eq:var} is the variance of $L(y,\xi)$, with respect to the
        `twisted' distribution,
        \[
        \frac{\exp\left( {t \frac{L(y,\xi)}{\alpha_y} }
            \right) dF(\xi)}{\mathbb E\left[ \exp\left({t \frac{L(y,\xi)}{\alpha_y}}\right) \right]},
        \]
        and the overall expression is strictly positive since
        $t^2/\a_y^3 > 0$. It follows that
        $\a_y\L\left(y,\frac{t}{\a_y}\right)$ is strictly convex in $\a_y$. Similarly, observe that $\alpha_x \Lambda\left(x,-\frac{t}{\alpha_x}\right)$ is strictly convex in $\alpha_x$. Consequently, it is straightforward to see that the Hessian
         of $J(\alpha_x,\alpha_y,t) \quad \forall t \in \mathbb{R}$
         is positive definite. Therefore, $J(\alpha_x,\alpha_y,t)$ is
         strictly convex $\forall~\alpha_x,\alpha_y~\in[0,1]\times
         [0,1]$. We also know that the cumulant generating function is
         convex in  $t$ in general, but it is strictly convex due to
         Assumption~\ref{assume:2}. Since the sum of two strictly convex function is strictly convex, $J(\alpha_x,\alpha_y,t)$ is strictly convex $\forall t \in \mathbb{R}$.
	
Next, observe that for any $\eta>0$, there exists a $t \in
\mathbb{R}$, such that for a given  $\alpha_x,\alpha_y \in[0,1]\times
[0,1]$, \[J(\alpha_x,\alpha_y,t ) \leq  -  I(\alpha_x,\alpha_y)  +
\eta.\] We now follow the arguments in Sec. 3.2.5 of ~\cite{BoVa2004}. For any $\beta \in [0,1]$ and $\a_x^j,\a_y^j \in
[0,1]\times [0,1]$, $j=1,2$,
	\begin{align*}
			- I(\beta \alpha^1_x +(1-\beta) \alpha^2_x,\beta \alpha^1_y + (1-\beta)\alpha^1_y ) &= \inf_{t \in \mathbb R } J(\beta \alpha^1_x +(1-\beta) \alpha^2_x,\beta \alpha^1_y + (1-\beta)\alpha^1_y,t ) 
			\\ 
			& \leq J(\beta \alpha^1_x +(1-\beta) \alpha^2_x,\beta \alpha^1_y + (1-\beta)\alpha^1_y,\beta t_1 +(1-\beta) t_2 ) 
			\\
			& < \beta J(\alpha^1_x , \alpha^1_y , t_1  ) + (1-\beta) J(\alpha^2_x , \alpha^2_y , t_2  ) 
			\\
			& \leq - \beta I(\alpha^1_x,\alpha^1_y)  - (1-\beta) I(\alpha^2_x,\alpha^2_y) + \eta,
		\end{align*}	
	where penultimate inequality follows from Jensen's
        inequality. Since $\eta$ is arbitrary, it follows that
        $I(\a_x,\a_y)$ is strictly concave.
	\end{proof}

\begin{lemma} \label{lem:concrate2}
	$\min_{x \in Q(\d)}  \sum _{y \in D} I(\alpha_x,\alpha_y) \mathbb{I}_{\{ f(y) < f(x) \}}$ is strictly concave in $\{\alpha_1,\alpha_2,\ldots \alpha_d\} \in[0,1]^d$.
\end{lemma}
\begin{proof}
		From Lemma~\ref{lem:concrate1} and the fact that the minimum of strictly concave functions preserves strict concavity, the proposition follows. 
	\end{proof}

\begin{theorem}
	The following constraint maximization problem is strictly concave,
	\begin{align}
	\label{eq:ola}
	\max_{\{\alpha_1,\alpha_2,\ldots \alpha_d\}}  \quad & \min_{x
		\in
		Q(\d)}
	\sum
	_{y \in
		D}
	I(\alpha_x,\alpha_y)
	\mathbb{I}_{\{
		f(y) <
		f(x)
		\}}, \\
	\nonumber
	\quad \text{such that}  \quad  \quad &\sum_{i=1}^d \alpha_i = 1, ~ \alpha_i \in [0,1] \forall i \in\{1,2,\ldots,d\}.	
	\end{align}
	\end{theorem}
\begin{proof}
	The proof immediately follows from Lemma~\ref{lem:concrate2}. 
	\end{proof}

Hence, the optimal allocation strategy is the solution of \eqref{eq:ola}. Next, we illustrate this optimization for specific cases.
\subsection*{Example 1: Normal Distribution}

Assume that $L(x,\xi) \sim \mathcal{N} (f(x),  \sigma^2(x))$, and observe that
\begin{align} 
  \nonumber
	I(\a_x,\a_y) &= \sup_{t \in R} \left( - \alpha_y \Lambda\left(y,\frac{t}{\alpha_y}\right)  - \alpha_x \Lambda\left(x,-\frac{t}{\alpha_x}\right)  \right) 
	\\
  \label{eq:gauss-loss}
	&= \frac{1}{2} (f(x)-f(y))^2 \left(\frac{\sigma^2(y)}{\alpha_y}+ \frac{\sigma^2(x)}{\alpha_x}\right)^{-1}.
	\end{align}
\\
Therefore, using Theorem~\ref{thrm:1}, the {\sc LDP} for a given $\epsilon > 0$ is
\begin{align}
\nonumber
\lim_{n\to \infty} \frac{1}{n} \log \mathbb{P}(  f(\hat{x}_n)- f(x^*) > \epsilon )
&= -\min_{x \in Q(\d)}  \huge\sum _{y \in D} \frac{1}{2} (f(x)-f(y))^2 \left(\frac{\sigma^2(y)}{\alpha_y}+ \frac{\sigma^2(x)}{\alpha_x}\right)^{-1}\mathbb{I}_ {\{f(x) > f(y)\}},
\end{align}
where  $Q(\d)$ is as defined in Theorem~\ref{thrm:1}.

\subsection*{Example 2: Binomial Distribution}
Assume that $L(x,\xi) \sim Bin (\frac{f(x)}{m}, m)$, where $m$ is the number of binomial trials. Observe that
\begin{align} 
\nonumber
I(\alpha_x,\alpha_y) &= \sup_{t \in R} \left( - \alpha_y \Lambda(y,\frac{t}{\alpha_y})  - \alpha_x \Lambda(x,-\frac{t}{\alpha_x})  \right) 
\\
\label{eq:bin-loss}
&= m \left( - \alpha_y \log \left(1-\frac{f(y)}{m} + \frac{f(y)}{m} e^{\frac{t^*}{\alpha_y}}\right)  - \alpha_x \log \left(1-\frac{f(x)}{m} + \frac{f(x)}{m} e^{-\frac{t^*}{\alpha_x}}\right)  \right),
\end{align}
\\
where $t^* = \log \left(  \frac{f(x)(m-f(y))}{f(y)(m-f(x))}  \right) \left[\frac{1}{\alpha_x} + \frac{1}{\alpha_y}\right]^{-1} $.Therefore using Theorem~\ref{thrm:1}, the {\sc LDP} for a given $\epsilon > 0$ is
\begin{align*}
\lim_{n\to \infty} \frac{1}{n} \log \mathbb{P}( \omega \in \Omega : f(\hat{x}_n)- f(x^*) > \epsilon )
&= - \min_{x \in Q(\delta)}  \sum _{y \in D} m \Bigg( - \alpha_y \log \left(1-\frac{f(y)}{m} + \frac{f(y)}{m} e^{\frac{t^*}{\alpha_y}}\right) \\
& - \alpha_x \log \left(1-\frac{f(x)}{m} + \frac{f(x)}{m} e^{-\frac{t^*}{\alpha_x}}\right)  \Bigg) \mathbb{I}_ {\{f(x) > f(y)\}},
\end{align*}
where  $Q(\delta)$ is as defined in Theorem~\ref{thrm:1}.

\section{Sequential Optimization}
The optimization problem~\eqref{eq:ola} is solved by the master
machine, and the sampling budget is assigned to the slave
machines. Observe that~\eqref{eq:ola} assumes that the master machine has complete knowledge of the true
cumulant generating function of the (stochastic) loss function. In
practice, of course, this is unknown and the master machine must rely
on empirical estimates of the objective from the slave machines. An
appropriate approach to solving the optimization problem would be to
perform a sequential optimization as sample estimates accumulate. With
a large, but finite budget the sequential optimum should be close to
the optimizer of~\eqref{eq:ola}.

We demonstrate the computation in two different scenarios. First, we
assume that the (rate) function $J(\a_x,\a_y,t)$ can be analytically
optimized over $t$. In this case, the objective in~\eqref{eq:ola} is
simpler to estimate and optimize. Note
that these instances are direct analogues in the {\sc SAA} context of the ranking and selection
(R\&S)
problems studied in~\cite{GlJu2004,HuPa2013}, and ~\cite{GlJu2018}. To deal with
these types of problems, we present Algorithm~\ref{algo:1} below, that parallels~\cite{HuPa2013} Algorithm 2.

Second, in many applications of {\sc SAA}, the optimization of $J(\cdot,\cdot,t)$
must be carried out numerically since closed forms are not
available. These instances are far more complicated than the
straightforward R\&S analogues considered above, as the geometry of
the loss function now plays a prominent role. 
Algorithm~\ref{algo:2} below exploits an {\it
	expectation-maximization} (EM) type iterative structure to
recursively compute the optimal allocation efficiently.

Of course, in either scenario, the rate functions and objectives must
be estimated by Monte Carlo sampling at the slave machines. Let $\hat J(\a_x,\a_y,t) := \a_y \hat \L\left(y,\frac{t}{\a_y}\right)
+ \a_x \hat\L\left(x,-\frac{t}{\a_x}\right)$, where
$\hat\L\left(x,\frac{t}{\a_x}\right) := (\lfloor n\a_x \rfloor)^{-1}
\sum_{i=1}^{\lfloor n\a_x \rfloor} \exp\left(\frac{t}{\a_x} L(x,\xi^x_i)\right)$ is
the natural empirical estimator of the log moment generating
function. We also define the set $\hat Q(\d) = \{x \in D : \hat f(x) >\hat
f(x^*) + \e\}$, for a given $\e > 0$. In the remainder of this
section, we assume a fixed $\e > 0$. Let $\hat \a_x(n)$
represent the estimated allocation at location $x \in D$ with total
sampling budget $n$. We define the {\it optimality gap} of the estimator as
\[
OG(n) := \sum_{x \in D} |\hat \a_x^{(n)} - \a_x|,
\]
where $\a_x$ is the `true' optimal linear allocation obtained by solving~\eqref{eq:ola}.

\subsection{Optimization with Closed-forms}
Consider situations where the `inner' optimization 
$\sup_{t\in \bbR} J(\a_x,\a_y,t)~x,y~\in~D$ can be completed in closed form analytically.
For instance, in Examples 1 and 2 above closed forms were derived for
cases where the loss functions at each of the design points are
Gaussian and binomially distributed (respectively). Let $\hat
I(\cdot,\cdot)$ represent the Monte Carlo estimate of this closed
form, which will require estimation of the mean and (possibly) the variance. Algorithm 1 proceeds iteratively by estimating the optimal allocation
while accumulating more and more samples at each iteration. This method parallels
Algorithm 2 in ~\cite{HuPa2013}. 

\begin{table}[h]
	\caption{Algorithm 1.}
	\centering
	\begin{tabular}{ | c | l |}
		\hline
		Step 0 & Initialize pilot sample $n_x^{(0)} = N^0$ at each $~x \in
		D$.\\
		For each $k \geq 0:$ & \\
		
		Step 1 &
		\begin{tabular}[t]{{@{}l@{}}}
			Generate $n_x^{(k)}$ i.i.d. samples at each $x \in D$.\\
			Compute $\hat I(\a_x,\a_y)~\forall~x,y~\in D$, $\hat
			f(x)~\forall x~\in D$ and $\hat Q(\d)$ using all
			$\sum_{i=0}^kn_x^{(i)}$ samples.\\
		\end{tabular}
		\\
		Step 2 &
		\begin{tabular}[t]{{@{}l@{}}}
			Compute $(\a_x^{(k+1)},~x \in D) = \underset{\{\a_x, x\in D\}} {\arg\max}~\underset{\{x
				\in \hat Q(\d)\}}{ \min} \sum_{y \in D} \hat I(\a_x,\a_y)
			\mathbb I_{\{\hat f(x) > \hat f(y)\}}$.\\
			Generate $T_j ~ \forall j=\{1,2,\ldots N\}$, where $T_j$ has empirical distribution with probability \\ $\alpha_x^{(k+1)}$ on support
			$x \in D$. Set $n_x^{(k+1)} = \sum_{j=1}^N \mathbb{I}_{\{T_j = x\}} ~\forall x \in D$.
		\end{tabular}
		\\
		Step 3 & Repeat Steps 1 and 2 , until  sampling budget exhausts.\\
		\hline
	\end{tabular}
	\label{algo:1}
\end{table}

Observe that the algorithm runs until the sampling budget is exhausted
with no guarantees on convergence to the true optimal allocation
rule. Consistency results from ~\cite{HuPa2013} and ~\cite{GlJu2004} imply that with
a large, but finite budget the allocation obtained at the end of the
procedure should closely match the optimal linear allocation rule.We
now illustrate the algorithm by running through a couple of examples.

\paragraph{Gaussian Loss:}
For simplicity we assume that the variances $\s^2(x)~ \forall x\in D$ are known, and the mean
value at each of the design points is estimated using the natural
estimator. 

The choice of $\d$ determines the error tolerance, and affects the
allocation budget. Figure~\ref{fig:1}(a) depicts a case where $\d$ is much smaller than the resolution of
the grid. In this case, both the true allocation and the estimated
allocation place much of the sampling effort near the optimizer. On
the other hand, when $\d$ is of the order of the grid resolution,
Figure~\ref{fig:1}(b) demonstrates that both the true and the
estimated allocations expend substantial sampling efforts near the
``boundary'' of the sets $Q(\d)$ and $\hat Q(\d)$ respectively. Figure~\ref{fig:1a} illustrates that in both the cases, optimality gap appears to converge, but with large variance.

\begin{figure}[!htb]
	\centering
	\begin{minipage}{.9\textwidth}
		\begin{subfigure}[h]{0.4\textwidth}
			\includegraphics[width=\textwidth,height=\textwidth,keepaspectratio]{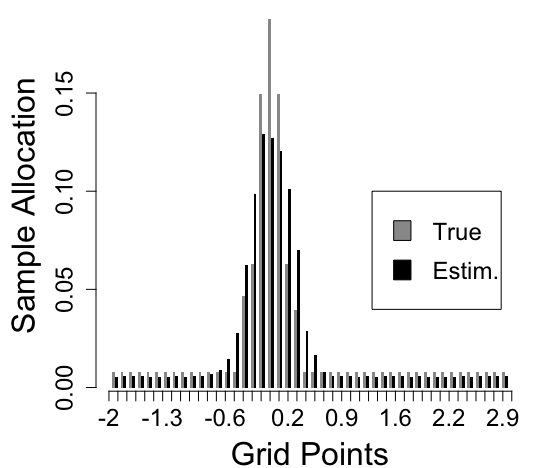}
			\caption{$\delta = 0.05$}
		\end{subfigure}
		~~~~~~~~~~~~~~~~
		\begin{subfigure}[h]{0.4\textwidth}
			\includegraphics[width=\textwidth,height=\textwidth,keepaspectratio]{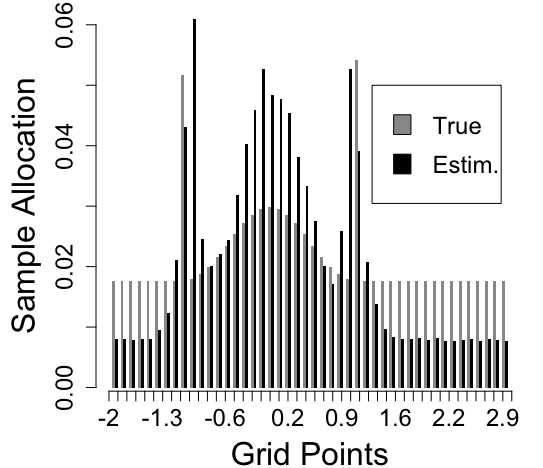}
			\caption{$\delta = 1.0$}
		\end{subfigure}
		\caption{Approximate and true allocations with Gaussian distributed
			loss functions. Total sampling budget $n=4600$, $|D| = 46$ and
			averaged over 50 sample paths.}
		\label{fig:1}
	\end{minipage}
	~
	\begin{minipage}{.9\textwidth}
		\begin{subfigure}[h]{0.4\textwidth}
			\includegraphics[width=\textwidth,height=\textwidth,keepaspectratio]{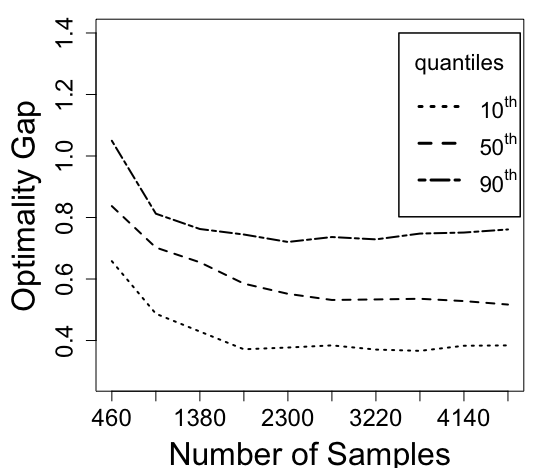}
			\caption{$\delta = 0.05$}
		\end{subfigure}
		~~~~~~~~~~~~~~~~
		\begin{subfigure}[h]{0.4\textwidth}
			\includegraphics[width=\textwidth,height=\textwidth,keepaspectratio]{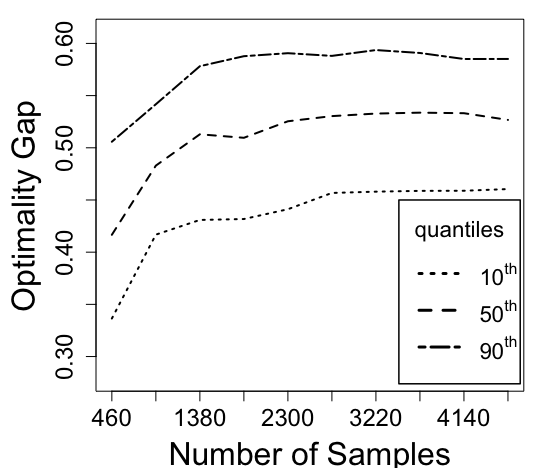}
			\caption{$\delta = 1.0$}
		\end{subfigure}
		\caption{Sample quantile of Optimality Gap for Gaussian distributed loss functions for 50 Sample paths. Total sampling budget $n=4600$ and $|D| = 46$.\\}
		\label{fig:1a}
	\end{minipage}
\end{figure}

\paragraph{Binomial Loss:} Next, in the case of the binomially distributed loss function, it suffices to compute the objective value $\hat f(x)$ using the natural, plug-in estimator for~\eqref{eq:bin-loss}. Our observations here parallel the Gaussian case.

\subsection{Optimization without Closed-forms}
It is rare to place explicit distributional assumptions on the loss
function at different design points in the grid, and typical stochastic programming models assume regularity conditions
on the loss function $L(\cdot,\cdot)$ and stochasticity conditions on
$\xi$. The distributional conditions are then consequences of these
two ingredients. In general, then, it is highly unlikely that there
exists a closed form for the optimization $\sup_{t \in \bbR}
J(\cdot,\cdot, t)$, and the optimization must be carried out
numerically on a Monte Carlo estimate $\hat J(\cdot,\cdot,t)$.




As noted before, Algorithm 1 has no guarantees on convergence within a fixed
number of iterations, since it is only running till the
sampling budget is exhausted. On the other hand, there are many applications of {\sc SAA} where it is
useful to run the algorithm till convergence. For instance, in
data-driven applications, it may be possible to obtain $n$ samples
repeatedly from a simulator, or by bootstrap sampling of a given
dataset. To handle such situations we propose a second iterative
algorithm that is expectation-maximization (EM) like, and proceeds in two iterative
steps. In step one, for a fixed linear allocation rule
$(\alpha_x,~x\in D)$, we compute $\hat J(\a_x,\a_y,t)$ for every
$x,y~\in~D$ and identify the
optimal $t(x,y)$. In step two, we
compute an allocation using the objective in~\ref{eq:ola}
albeit with $I(\a_x,\a_y, t(x,y))$ from step one. These two steps are
iterated till there is no improvement in the allocation in step
two. The following display summarizes the algorithm. 

\begin{table}[h]
	\caption{Algorithm 2.}
	\centering
	\begin{tabular}{ | c | l |}
		\hline
		Step 0 & Fix $(\a^{(0)}_x,~x\in D)$, where $\a^{(0)}_x
		\in [0,1]$ and $\sum_x \a^{(0)}_x = 1$. \\
		For each $k \geq 0$: & \\
		Step 1 & \begin{tabular}[t]{@{}l@{}} 
			Generate $T_j ~ \forall j=\{1,2,\ldots n\}$, where $T_j$ has empirical distribution  with probability \\ $\alpha_x^{(k)}$ on support
			$x \in D$. Set $n_x^{(k)} = \sum_{j=1}^n \mathbb{I}_{\{T_j = x\}} ~\forall x \in D$.\\
			Generate $n_x^{(k)}$ i.i.d. random samples at each $~x\in D$. \\ 
			Compute $t^{(k)}(x,y) = \arg\sup_{t\in\bbR} \hat J(\a_x^{(k)},\a_y^{(k)},t)$ for all $x,y \in D$, and\\
			Compute $\hat f(x)$ for all $x \in D$ and $\hat Q(\d)$ using all $\sum_{i=0}^k n_x^{(i)}$ samples. \end{tabular} \\
		Step 2 & Compute $(\a^{(k+1)}_x,~x\in D) \in \underset{\{\a_x,~x\in D \}}{\arg\max}~ \underset{{x
				\in \hat Q(\d)}}{\min} \sum_{y\in D} \hat I(\a_x,\a_y,t^{(k)}(x,y))
		\mathbb I_{\{\hat f(x) > \hat f(y)\}}$.\\
		Step 3 & Repeat Step 1 and 2, until $(\a_x^{(k)},~x\in D)$
		converges.\\
		\hline
	\end{tabular}
	\label{algo:2}
\end{table}

Note that we assume that the sampling budget $n$ is fixed, and allow the
algorithm to produce $n$ samples on each iteration. It is possible to
couple the iterative scheme to the sampling budget by fixing the
number of iterations {\it a priori} to $K=\lfloor n*\gamma \rfloor$,
where $\g \in [ 0,1]$ is fixed. At the $k$th iteration, $\sum_{i=0}^k\a_x^{(i)} K$
samples is generated at sample point $x \in D$. In this case, the
algorithm terminates once $K$ iterations have been completed.

\paragraph{Mean-squared Error} Consider a squared error loss, widely used in empirical
risk minimization of machine learning models, where $L(x,\xi) := (x -
\xi)^2$. For simplicity, we assume that the design space is one
dimensional and that $\xi$ is Gaussian.

%

Observe that the estimated allocation is quite close to the true
allocation, even in this case (see Figure~\ref{fig:3}). Significant budget allocations are once
again made in the vicinity of the boundary of
$Q(\d)$. Figure~\ref{fig:3a} on the other hand, amply demonstrates
that while the estimators are consistent with a large budget ``on
average'' as demonstrated by the 50th percentile lines, there is
significant variance in the optimality gap even at large sample
values as shown by the spread between the 90th and 10th percentiles.

\begin{figure}[!htb]
	\centering
	\begin{minipage}{.9\textwidth}
		\begin{subfigure}[h]{0.4\textwidth}
			\includegraphics[width=\textwidth,height=\textwidth,keepaspectratio]{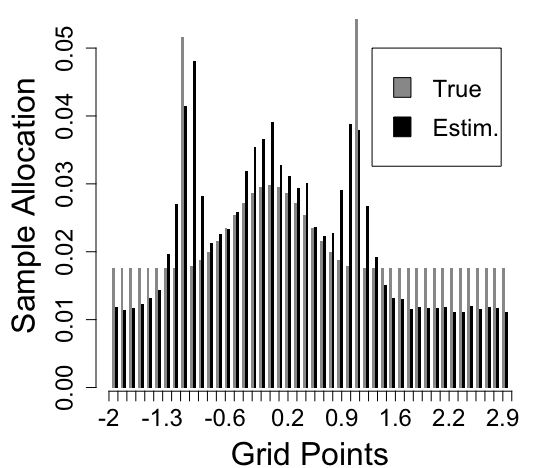}
			\caption{$\delta = 1.0$}
		\end{subfigure}
		~~~~~~~~~~~~~~~~
		\begin{subfigure}[h]{0.4\textwidth}
			\includegraphics[width=\textwidth,height=\textwidth,keepaspectratio]{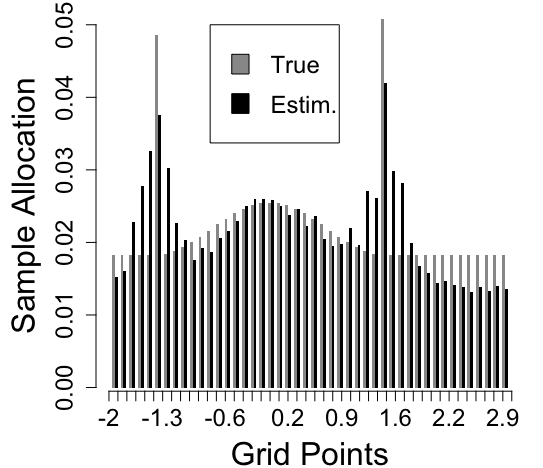}
			\caption{$\delta = 2.0$}
		\end{subfigure}
		\caption{Approximate and true allocations for squared loss with Gaussian random samples.  Sample at each iteration $n=460$, $|D| = 46$ and averaged over 30 sample paths.}
		\label{fig:3}
	\end{minipage}
	~
	\begin{minipage}{.9\textwidth}
		\begin{subfigure}[h]{0.4\textwidth}
			\includegraphics[width=\textwidth,height=\textwidth,keepaspectratio]{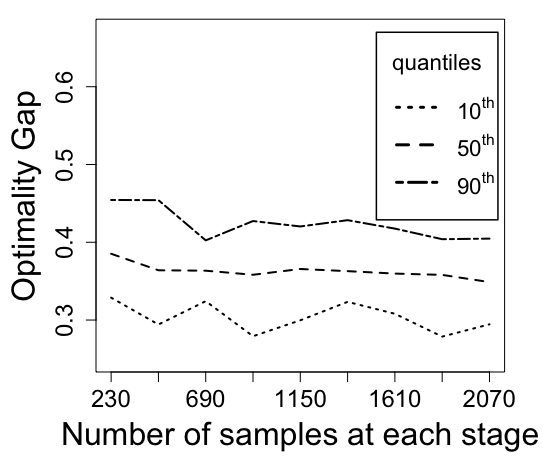}
			\caption{$\delta = 1.0$}
		\end{subfigure}
		~~~~~~~~~~~~~~~~
		\begin{subfigure}[h]{0.4\textwidth}
			\includegraphics[width=\textwidth,height=\textwidth,keepaspectratio]{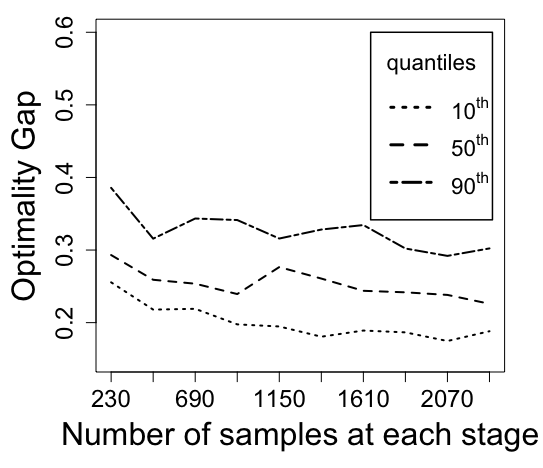}
			\caption{$\delta = 2.0$}
		\end{subfigure}
		\caption{Sample quantile of Optimality Gap for squared
			loss functions with Gaussian random samples for 30 Sample paths and $|D| = 46$.}
		\label{fig:3a}
	\end{minipage}

\end{figure}

\section{Conclusions and Future Directions}
We study the problem of optimally allocating a sampling budget in
order to compute a sample average approximation ({\sc SAA}) of the solution of a single-stage
stochastic program. Under a fixed finite discretization of the design
space (or `grid'), we first establish a large deviations principle satisfied by
the regret, defined as $f(\hat x_n) - f(x^*)$, where $f(\cdot)$ is the
true objective, $\hat x_n$ is the {\sc SAA} estimate of the optimizer and
$x^*$ the true optimizer. Next, we identify a constraint maximization
problem, whose solution identifies an optimal linear allocation rule
that maximizes the decay rate of the likelihood of identifying an incorrect
optimal design point, in the limit of a large sampling
budget. Finally, we designed two different algorithms to sequentially
implement this optimization.

The developments in this paper lead to multiple important and open
problems, relevant to both simulation optimization and machine
learning more broadly. First, our current treatment of the regret
effectively assumes that the grid is fixed. An important question is
how the grid size affects the large deviations rate function. In
particular, it can be easily seen that when the grid has the
cardinality of the continuum, the rate function does not exist. On the
other hand, there is definite benefit in scaling the grid size with
the sampling budget. How should this be done to obtain a large
deviations principle in the limit?

Second, while results in ~\cite{GlJu2004} and ~\cite{HuPa2013} can be
straightforwardly adapted to establish consistency of the sequentially
estimated allocation rule, the efficiency of the estimator is
unknown. In particular, we conjecture that the rate function estimators used are highly
inefficient. This follows from the fact that we use the canonical estimator for the cumulant
generating function, and it is conjectured that the latter estimators are
heavy-tailed ~\cite{GlJu2011}. On the other hand, we are really only interested in the
accuracy of the estimated objective in the vicinity of the true
optimizer, and not the global accuracy. Closer to the optimizer, and
in the large sampling budget limit, we conjecture that it is possible
to use fewer moments to accurately estimate the rate function, leading
to substantial improvements in efficiency.

Third, note that the algorithms designed here are not
dimension free, and we conjecture that even with strictly convex
objective functions $f$ they will not scale well. We
postulate that it is possible to combine~\eqref{eq:ola} with multiple
stochastic gradient descent (SGD) crawlers starting at each of the
grid-points, and letting these iterate a fixed number of times, to
make the allocation optimization algorithm dimension-free.

Fourth, `gridding' the design space has significant algorithmic
advantages since it allows our present algorithms to be implemented in
a `divide and conquer' manner, whereby slave machines sample and estimate
moments of the loss function at each design point, and communicate
this to the master machine for optimization. On the other hand,
estimates at a given design point must be close to those at other
design points in its
neighborhood. With a fine grid, the amount of communication overhead
required to implement a fully parallel computation is likely
significant. We speculate that it should be possible to design a
completely decentralized optimization scheme to determine the optimal
allocation just by performing local message passing. Such a scheme
would yield substantial reductions in communication overhead.

\bibliographystyle{plain}
\bibliography{refs}

\begin{thebibliography}{1}

\bibitem{BoVa2004}
Stephen Boyd and Lieven Vandenberghe.
\newblock {\em {C}onvex {O}ptimization}.
\newblock {Cambridge University Press}, New York, 2004.

\bibitem{Chen2000}
Chun-Hung Chen, Jianwu Lin, Enver Y{\"u}cesan, and Stephen~E. Chick.
\newblock Simulation budget allocation for further enhancing the efficiency of
  ordinal optimization.
\newblock {\em Discrete Event Dynamic Systems}, 10(3):251--270, 2000.

\bibitem{DeZe2010}
Amir Dembo and Ofer Zeitouni.
\newblock {\em Large Deviations Techniques and Applications, volume 38 of
  Stochastic Modelling and Applied Probability}.
\newblock Springer-Verlag, Berlin, 2nd edition, 2010.

\bibitem{GlJu2004}
Peter Glynn and Sandeep Juneja.
\newblock A large deviations perspective on ordinal optimization.
\newblock In R.~G.~Ingalls et~al., editor, {\em Proceedings of the 2004 Winter
  Simulation Conference}, pages 577--585, Piscataway, New Jersey, 2004. IEEE.

\bibitem{GlJu2011}
Peter Glynn and Sandeep Juneja.
\newblock Ordinal optimization: A nonparametric framework.
\newblock In S.~Jain et~al., editor, {\em Proceedings of the 2011 Winter
  Simulation Conference}, pages 4062--4069, Piscataway, New Jersey, 2011. IEEE.

\bibitem{GlJu2018}
Peter Glynn and Sandeep Juneja.
\newblock Selecting the best system, large deviations, and multi-armed bandits.
\newblock {\em arXiv preprint arXiv: 1507.04564v2}, 2018.

\bibitem{HuPa2013}
Susan~R Hunter and Raghu Pasupathy.
\newblock Optimal sampling laws for stochastically constrained simulation
  optimization on finite sets.
\newblock {\em INFORMS Journal on Computing}, 25(3):527--542, 2013.

\bibitem{ShDeRu2009}
Alexander Shapiro, Darinka Dentcheva, and Andrzej Ruszczy{\'n}ski.
\newblock {\em Lectures on Stochastic Programming: Modeling and Theory}.
\newblock SIAM, Philadelphia, Pennsylvania, 2nd edition, 2009.

\end{thebibliography}

\end{document}